\documentclass[aps,prl,notitlepage,superscriptaddress,twocolumn,amsmath,amssymb,10pt]{revtex4-1}
\usepackage{times,bbm}

\usepackage[usenames,dvipsnames]{xcolor}
\usepackage{mathtools}
\usepackage{tikz}  
\usetikzlibrary{arrows,shapes,positioning,shadows,backgrounds,fit}

\usepackage[english]{babel}
\usepackage{graphicx, dcolumn, bm, color, amsmath,amsthm, relsize, amsmath, dsfont, mathrsfs, empheq, verbatim, upgreek, etoolbox,braket}
\usepackage[caption = false]{subfig}

\usepackage[colorlinks=true,linkcolor=blue,urlcolor=blue,citecolor=blue]{hyperref}
\usepackage{graphics}
\usepackage{bm}
\usepackage{color}
\usepackage{amscd}
\usepackage{amsfonts}
\usepackage{amssymb}
\usepackage{graphicx}
\usepackage{tabularx}

\newcommand{\mean}[1]{\langle #1 \rangle}
\providecommand{\bgreek}[1]{\mbox{\boldmath$#1$}}

\newcommand{\Ham}{\mathcal{H}}
\newcommand{\Id}{\mathds{1}}

\def\one{\mathbbm{1}}

\usepackage{amsthm}
\theoremstyle{definition}

\theoremstyle{theorem}
\newtheorem{theorem}{Theorem}
\newtheorem*{theorem*}{Theorem}
\newtheorem{corollary}{Corollary}
\newtheorem*{corollary*}{Corollary}

\newtheorem*{lemma*}{Lemma}

\begin{document}

\newdimen\origiwspc%
\newdimen\origiwstr%
\preprint{}

\title{Time periodicity from randomness in quantum systems}

\author{Giacomo Guarnieri}
\email{giacomo.guarnieri@fu-berlin.de}
\affiliation{Dahlem Center for Complex Quantum Systems, Freie Universit\"{a}t Berlin, 14195 Berlin, Germany}
\affiliation{School of Physics, Trinity College Dublin, College Green, Dublin 2, Ireland}

\author{Mark T. Mitchison}
\affiliation{School of Physics, Trinity College Dublin, College Green, Dublin 2, Ireland}

\author{Archak Purkayastha}
\affiliation{School of Physics, Trinity College Dublin, College Green, Dublin 2, Ireland}

\author{Dieter Jaksch}
\affiliation{Clarendon Laboratory, University of Oxford, Parks Road, Oxford OX1 3PU, United Kingdom}
\affiliation{Centre for Quantum Technologies, National University of Singapore, 117543 Singapore}

\author{Berislav Bu\v{c}a}
\affiliation{Clarendon Laboratory, University of Oxford, Parks Road, Oxford OX1 3PU, United Kingdom}

\author{John Goold}
\affiliation{School of Physics, Trinity College Dublin, College Green, Dublin 2, Ireland}

\begin{abstract}
Many complex systems can spontaneously oscillate under non-periodic forcing. Such self-oscillators are commonplace in biological and technological assemblies where temporal periodicity is needed, such as the beating of a human heart or the vibration of a cello string. While self-oscillation is well understood in classical non-linear systems and their quantized counterparts, the spontaneous emergence of periodicity in quantum systems without a semi-classical limit is more elusive. Here, we show that this behavior can emerge within the repeated-interaction description of open quantum systems. Specifically, we consider a many-body quantum system that undergoes dissipation due to sequential coupling with auxiliary systems at random times. We develop dynamical symmetry conditions that guarantee an oscillatory long-time state in this setting. Our rigorous results are illustrated with specific spin models, which could be implemented in trapped-ion quantum simulators.
\end{abstract}

\date{\today}
\maketitle

Periodic dynamics is ubiquitous in our everyday experience and forms a convenient basis for understanding more complex time-dependent phenomena. Self-oscillators are an especially important class of oscillatory system, which spontaneously oscillate under non-periodic forcing~\cite{Jenkins_2013}. A familiar example is the harmonic tone produced when a steady stream of air is blown across the top of a glass bottle~\cite{Green2006}. This emergent regularity makes self-oscillators useful for technologies based on periodic motion, such as reciprocating engines and clocks. Recent years have seen a growth of interest in miniaturising such devices to the extreme limit where only a few quantum degrees of freedom are involved~\cite{Tonner2005,Gelbwaser2014,Roulet2017, Erker2017,Woods2018,Lindenfels2019,Schwarzhans2021}. This naturally raises the question of how self-oscillations can emerge from basic principles of quantum dynamics.

A standard approach to this problem is to start from a classical model of a self-oscillator and then quantize it. Perhaps the most prominent example is the van der Pol oscillator, which incorporates the non-linearity and dissipation necessary for self-oscillations in the classical domain. This model describes non-linear electrical circuits~\cite{Pol1934} and the semi-classical dynamics of a laser near threshold~\cite{Haken1975}, for instance. The quantized van der Pol oscillator has been found to exhibit a rich variety of behaviors including criticality~\cite{Ishibashi2017,Dutta2019} and quantum synchronisation~\cite{Lee2013,Loerch2016}, replicating known properties of the corresponding classical model. However, it remains unclear whether the rich physics of self-oscillation can arise in a quantum system without a clear classical analogue.

In this work, we demonstrate that periodic dynamics can emerge in quantum systems that are driven by dissipation processes distributed randomly in time. We describe how this self-oscillatory behavior can arise in quantum many-body systems, including those with a finite-dimensional Hilbert space and no semi-classical limit. Our approach is based on the so-called repeated-interaction scheme, or collision model~\cite{RauPR,Scarani2002,BuzekPRA,Campbell2021review}, where a stream of auxiliary systems interact sequentially with the system of interest. In general, collision models provide a versatile description of dissipation that reduces to a standard Lindblad equation in an appropriate limit of fast collisions~\cite{GiovannettiPRL2012,VacchiniPRL2016,LorenzoPRA2017}. Here, we consider the case where the collisions are separated by finite but \textit{random} time intervals~\cite{seah2019nonequilibrium}.

In this context, we provide precise conditions that guarantee the spontaneous emergence of periodic dynamics at long times. Our rigorous proof extends the concept of dynamical symmetries --- extensive or local algebraic conditions that have been recently studied in relation to Hamiltonian and Lindblad dynamics~\cite{buvca2019non,Marko1,buca2021algebraic,chinzei2020time,Cameron,scarsdynsym1,scardynsym2} --- to general linear quantum evolutions (quantum channels). Remarkably, we show that periodicity appears both at the level of individual random realizations and in the ensemble average. 
We illustrate our results with the example of a four sites XXZ spin ring with an excitation sink on one site~\cite{Dutta2020,Dolgirev2020,KollathQuantumWire2,Froeml2020,Popkov,Esslingerlocal1,Kuhr,Ortocat,DBA,alba2021noninteracting,Jamir3}. We argue that the self-oscillatory behavior of this model is generic and not a product of fine-tuning by demonstrating full non-perturbative stability to a wide set of generic external dissipative perturbations. We also discuss how the same physics may be observed in a long-range Ising model, as realized in recent experiments~\cite{Kim2009,Islam2013,Jurcevic2014,Jurcevic2015}. We work in units with $\hbar=1$ throughout.

\begin{figure}[b]
\includegraphics[width=\linewidth]{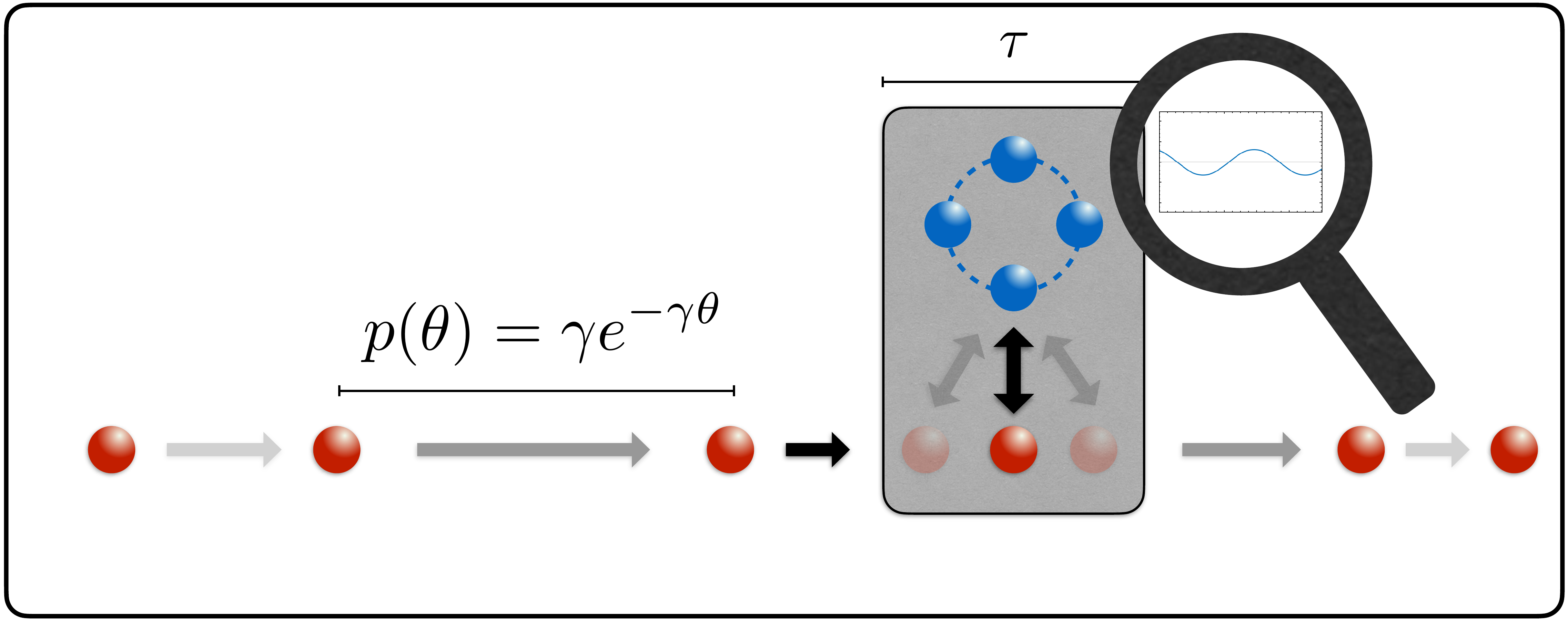}
\caption{Schematic representation of the randomized collision model, where a many-body system interacts with an ancilla for a fixed interaction time $\tau$ and where two consecutive interactions happen at a random time $\theta$.
}
\label{figschematic}
\end{figure}

\textit{Onset of time periodicity from randomized collision models.---}Consider a system $S$ with Hamiltonian $\hat{\Ham}_S$, initially prepared in a generic many-body state $\hat{\rho}_S^{(0)}$, and an environment consisting of a large number of non-interacting copies of another quantum system called the `ancilla' $A$. The free evolution of the system under $\hat{\Ham}_S$ is occasionally interrupted by an interaction with one of the ancillae over a time interval $\tau$, which is described by a joint unitary transformation $\hat{U}(\tau)$ acting on $S+A$ (see Fig.~\ref{figschematic}). Discarding the ancillae leads to an effective dissipative evolution for the system alone. The main advantage of this construction compared to standard models of open quantum systems theory is the replacement of a complex environment with a series of small ancillae, such that the consequent dynamics is analytically tractable but produces the same physics. Collision models have been recently employed to investigate the thermodynamics of open quantum systems under very general conditions~\cite{BarraSciRep,GabrieleNJP2018,guarnieri2020non} and can be generalised to include dynamical effects such as non-Markovianity~\cite{vacchini2016generalized,lorenzo2017quantum} or structured reservoir spectra~\cite{purkayastha2020periodically}.

Crucially, while each system-ancilla interaction is assumed to have a fixed duration $\tau$, the free evolution time $\theta$ in between two consecutive collisions is taken to be a random variable governed by a probability distribution $p(\theta)$. This scenario, which generalises the usual repeated-interaction scheme where $\theta$ is deterministic, was recently introduced in Ref.~\cite{seah2019nonequilibrium} to model the uncontrollable degrees of freedom of a thermal environment. Following the above picture, the dynamics after $n$ collisions have taken place leads to the state
\begin{equation}\label{totalmap}
    \hat{\rho}_S^{(n)} = \underbrace{\mathcal{U}_{S,\theta_n}\circ\Lambda_\tau\circ\mathcal{U}_{S,\theta_{n-1}}\circ\Lambda_{\tau}\circ\ldots\circ\mathcal{U}_{S,\theta_1}\circ\Lambda_\tau[}_{\text{n times}}\hat{\rho}_S^{(0)}],
\end{equation}
where $\theta_1,\ldots,\theta_n$ are $n$ possible outcomes of the random variable $\theta$, $\mathcal{U}_{S,\theta}[\bullet] = \hat{U}_S(\theta)\bullet\hat{U}^{\dagger}_S(\theta)$ denotes the free system evolution, and $\Lambda_{\tau} [\bullet] = \sum_k \hat{\Omega}_{k}(\tau) \bullet\hat{\Omega}_{k}^{\dagger}(\tau) $ represents the completely positive and trace-preserving (CPTP) map that describes a single collision, with Kraus operators $\hat{\Omega}_k(\tau) \equiv \hat{\Omega}_{\alpha\beta,\tau} = \sqrt{p_\alpha} \bra{\beta} \hat{U}(\tau) \ket{\alpha}$ and $\hat{\rho}_A = \sum_\alpha p_\alpha\ket{\alpha}\bra{\alpha}$. 

Let us now define the composite map $\tilde{\Lambda}_{\tau,\theta}[\bullet] \equiv \mathcal{U}_{S,\theta}\circ\Lambda_\tau[\bullet]$ and denote with a subscript $I$ the interaction picture with respect to the system's free evolution.
Due to its being CPTP, all the eigenvalues of $\tilde{\Lambda}_{\tau,\theta}$ lie inside or on the unit circle in the complex plane and there is always at least one right eigenvector with eigenvalue 1 (the stationary state) \cite{Wolf_2012,bruzda2009random,riera2020time}, which we denote by $\hat{\omega}^D$, such that $\tilde{\Lambda}_{\tau,\theta}\left[\hat{\omega}^D\right] = \hat{\omega}^D$. If there exist other eigenvalues with unit modulus then the corresponding eigenvectors oscillate under the action of $\tilde{\Lambda}_{\tau,\theta}$ and do not decay. These eigenvectors and corresponding eigenvalues, which lie on the unit circle, are referred to as the peripheral spectrum of $\tilde{\Lambda}_{\tau,\theta}$.
The following Theorem, which represents our main result, provides a set of precise conditions that guarantee the existence of such asymptotic oscillating states given a general quantum channel, yielding furthermore a direct way to explicitly construct them from the stationary state $\hat{\omega}^D$ (see Section B of Supplementary Material~\cite{SM} for the proof).

\begin{theorem}\label{theorem1}
Consider the CPTP map $\tilde{\Lambda}_{\tau,\theta} [\bullet] = \mathcal{U}_{S,\theta}\circ\Lambda_\tau[\bullet]$ and let $\hat{\omega}^D$ be its invariant state such that $\tilde{\Lambda}_{\tau,\theta}[\hat{\omega}^D] = \hat{\omega}^D$.
If there exists a system operator $\hat{\Xi}$ such that the following two conditions are satisfied
\begin{equation}
\mathrm{\mathbf{(i)}}\, \left[\hat{\Ham}_S, \hat{\Xi}\right] = \lambda \,\hat{\Xi}, \qquad \mathrm{\mathbf{(ii)}}\, \left[\hat{\Omega}_{k,I}(\tau), \hat{\Xi}\right] \hat{\omega}^D = 0,\,\forall k,\tau
\end{equation}
with $\lambda \in \mathbb{R}$, then the operator $\hat{\Xi}\,\hat{\omega}^D$ will evolve according to
\begin{equation}\label{theorem:Evol}
\tilde{\Lambda}_{\tau,\theta}[\hat{\Xi} \hat{\omega}^D] = e^{-i\lambda (\tau+\theta)} \,\hat{\Xi}\,\hat{\omega}^D.
\end{equation}
\end{theorem}
Note that Theorem~\ref{theorem1} holds for any dynamics described by a CPTP map, and thus is valid beyond the randomized collision model considered here. Physically, conditions (i) and (ii) characterize the operator $\hat{\Xi}$ as a \textit{generalized dynamical symmetry}~\footnote{Formally, the operator $\hat{\Xi}$ thus corresponds to a generalized rotation on the unit circle of the subspace of invariant states. It is possible to prove~\cite{lindblad1999general} that condition (ii) in fact is satisfied if $\hat{\Xi}$ belongs to a matrix sub-algebra of invariant operators of the adjoint map $\tilde{\Lambda}_{\tau,\theta}$, i.e. $\tilde{\Lambda}^{\dagger}[\hat{\Xi}] = \hat{\Xi}$ which defines dynamical symmetries and conserved quantities of the dynamical map. The converse is also true.}. In particular, condition~(i) defines a dynamical symmetry of the system's autonomous evolution~\cite{buvca2019non} while condition (ii)~expresses the requirement that this symmetry must be insensitive to the dissipation. Verifying the latter condition can be quite demanding in the presence of a generic environment; however, this complexity is substantially reduced in the case of collision models due to the simplicity of the ancillae. The oscillation frequency $\lambda$ is clearly unrelated to any timescale of the system-environment interaction, e.g.~$\tau,\theta$, and solely depends on the spectrum of $\hat{\Ham}_S$. We finally point out that, due to the symmetry of the spectrum under complex conjugation of any generic CPTP map, if a dynamical symmetry $\hat{\Xi}$ relative to eigenfrequency $\lambda$ exists, also $\hat{\Xi}^{\dagger}$ is a dynamical symmetry relative to $\lambda^*$.

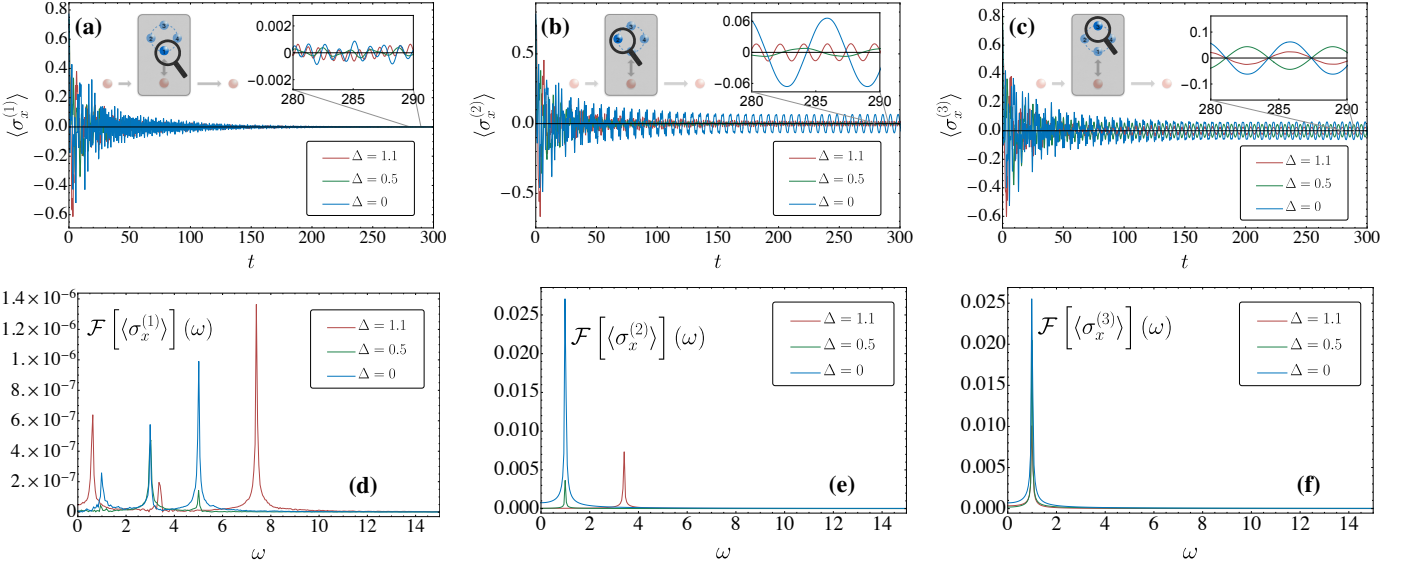
\begin{figure*}[htbp!]
\begin{center}
\hspace*{-1cm}
\begin{tikzpicture} 
  \node (img1)  {\includegraphics[width=0.33\linewidth]{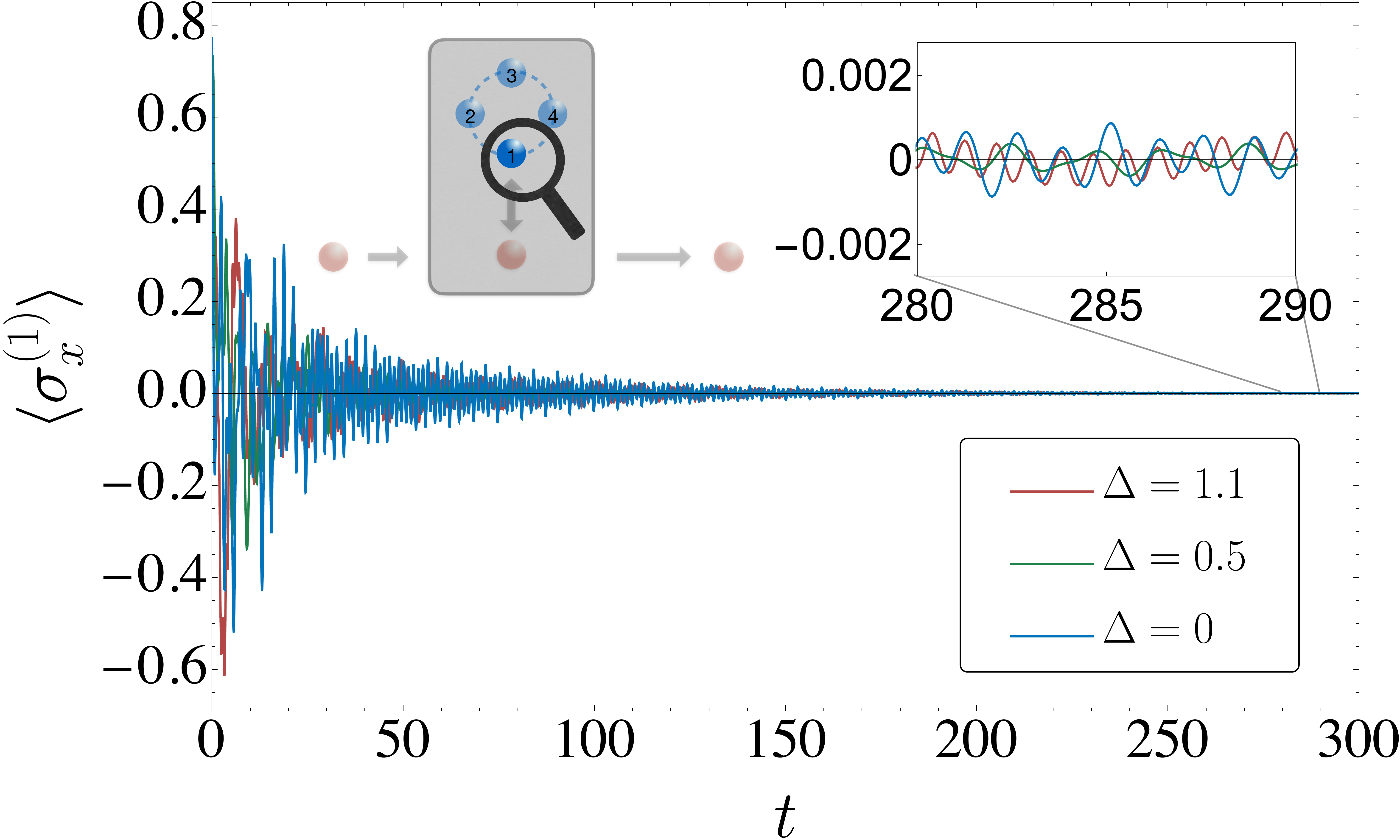}};
    \node[above=of img1, node distance=0cm, yshift=-1.7cm,xshift=-1.8cm] {{\color{black}{\bf{(a)}}}};
\end{tikzpicture}
\begin{tikzpicture} 
  \node (img2)  {\includegraphics[width=0.33\linewidth]{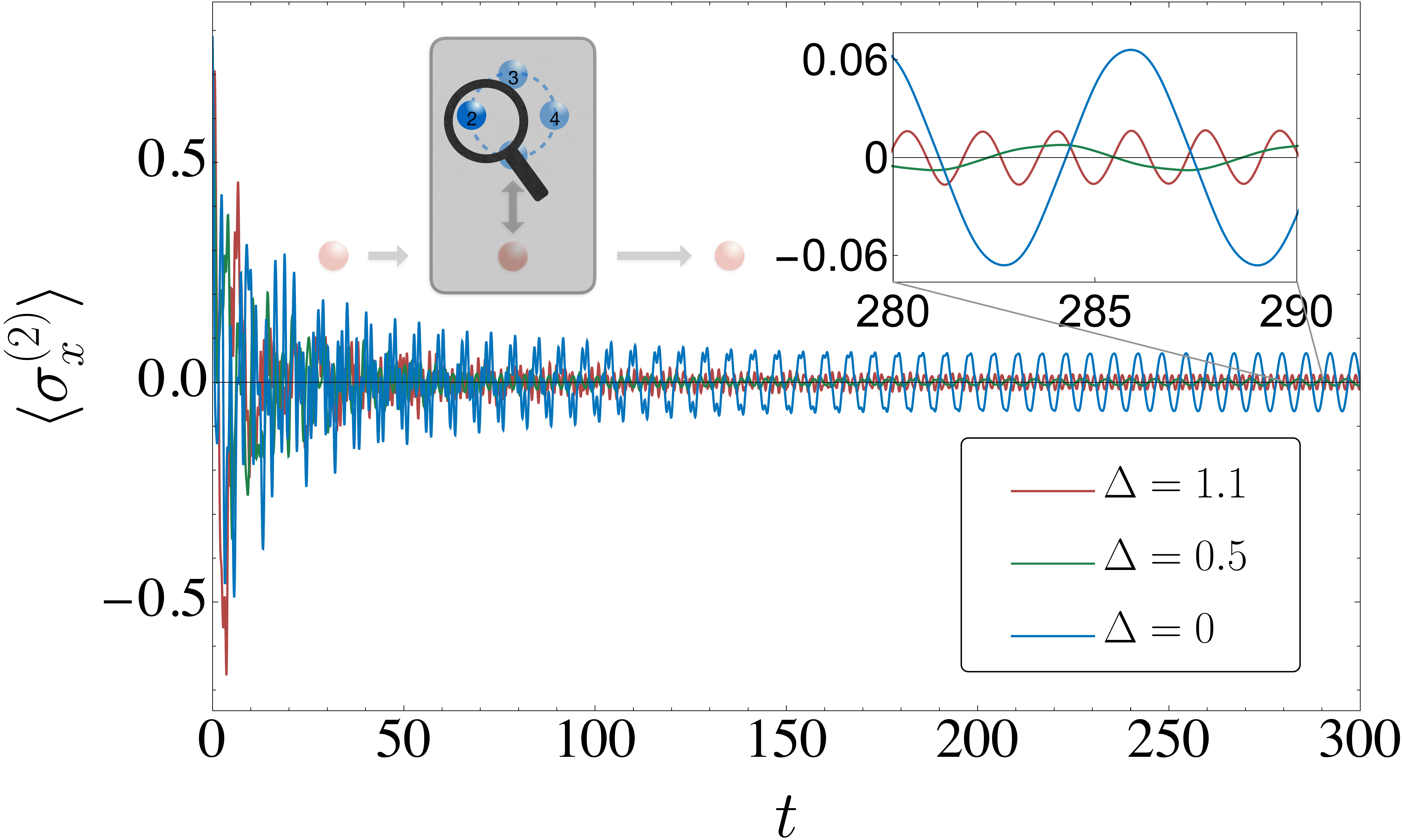}};
    \node[above=of img2, node distance=0cm, yshift=-1.7cm,xshift=-1.8cm] {{\color{black}{\bf{(b)}}}};
\end{tikzpicture}
\begin{tikzpicture} 
  \node (img3)  {\includegraphics[width=0.33\linewidth]{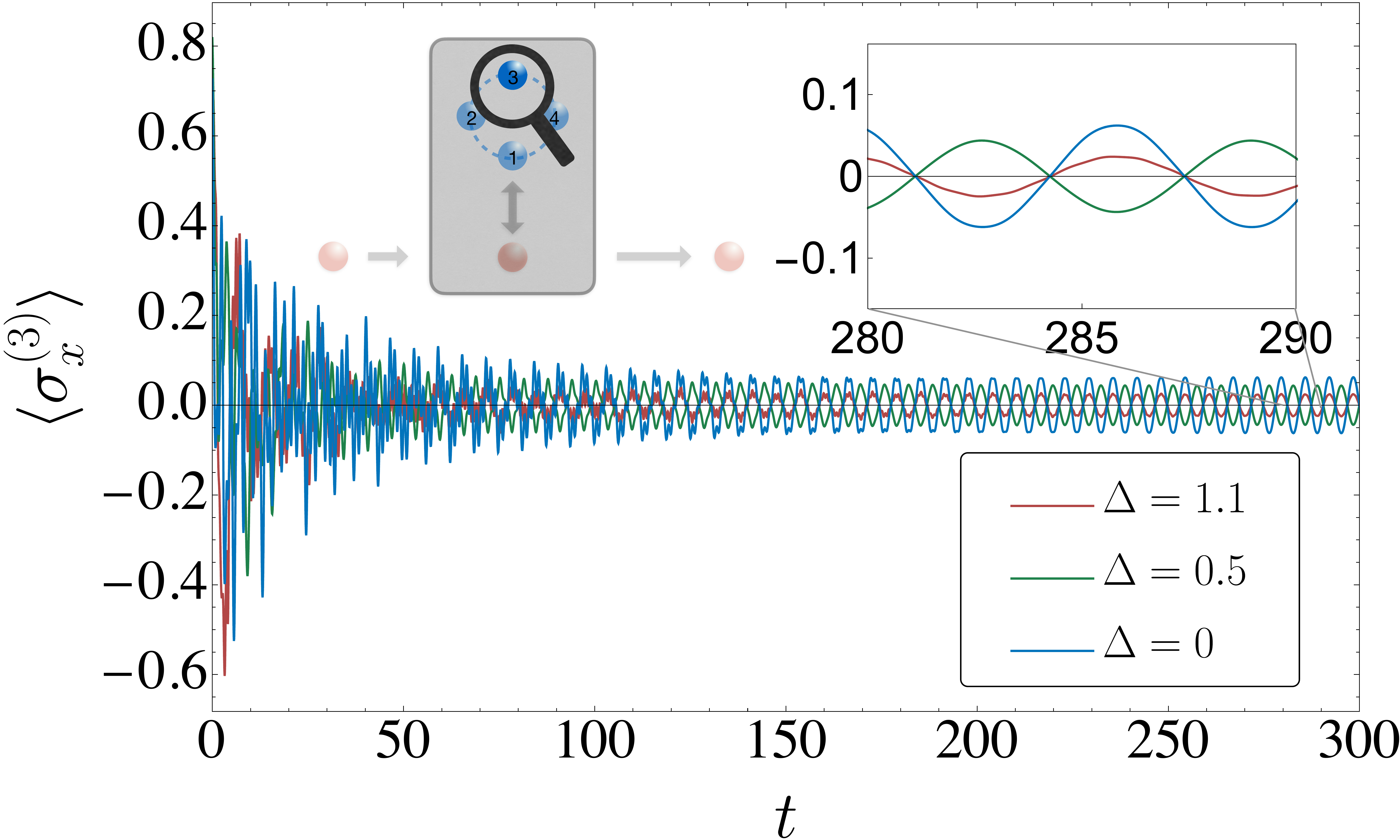}};
    \node[above=of img3, node distance=0cm, yshift=-1.7cm,xshift=-1.8cm] {{\color{black}{\bf{(c)}}}};
\end{tikzpicture}\\
\hspace*{-1.5cm}
\begin{tikzpicture} 
  \node (img4)  {\includegraphics[width=0.348\linewidth]{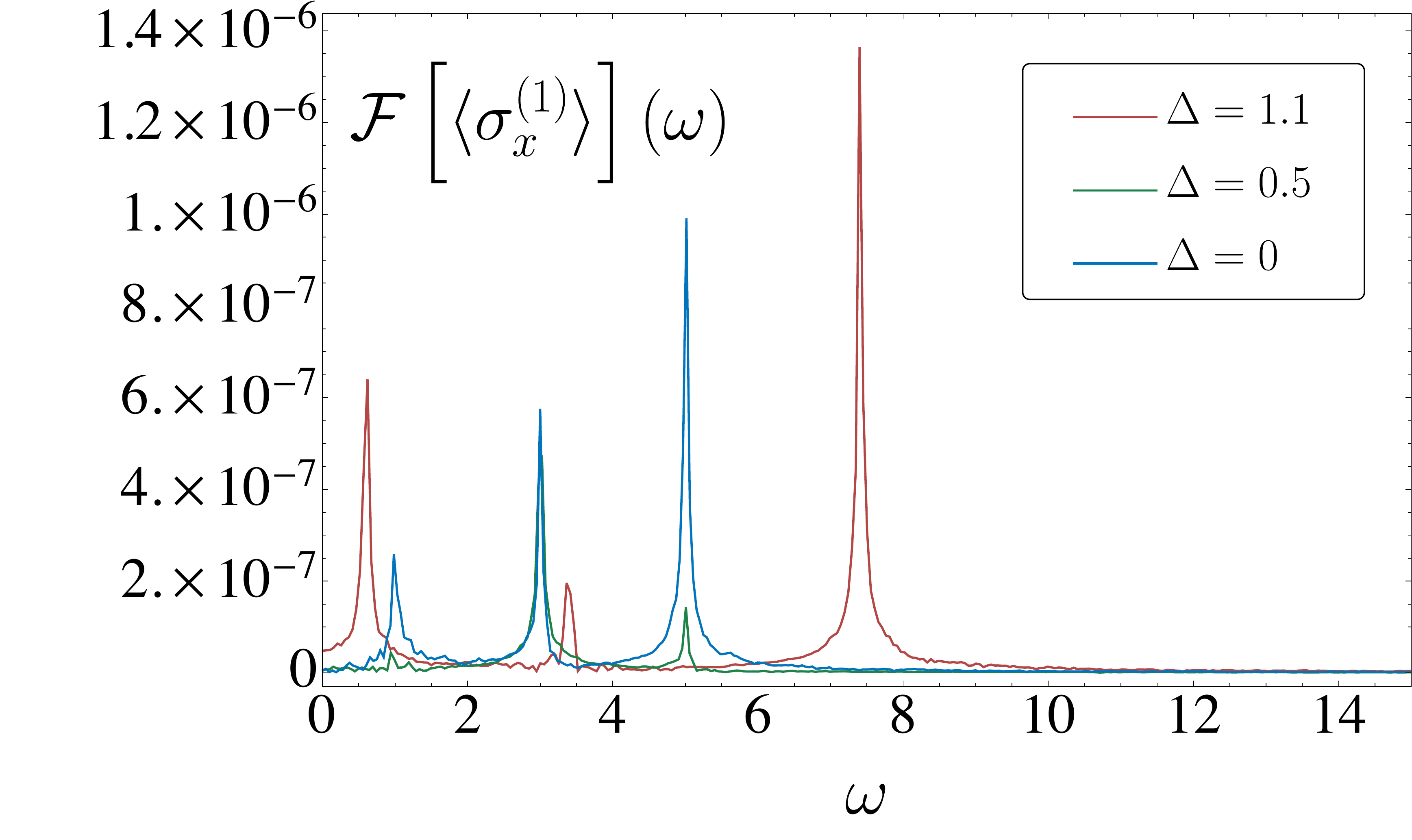}};
    \node[above=of img4, node distance=0cm, yshift=-4cm,xshift=2.1cm] {{\color{black}{\bf{(d)}}}};
\end{tikzpicture}
\begin{tikzpicture} 
  \node (img5)  {\includegraphics[width=0.33\linewidth]{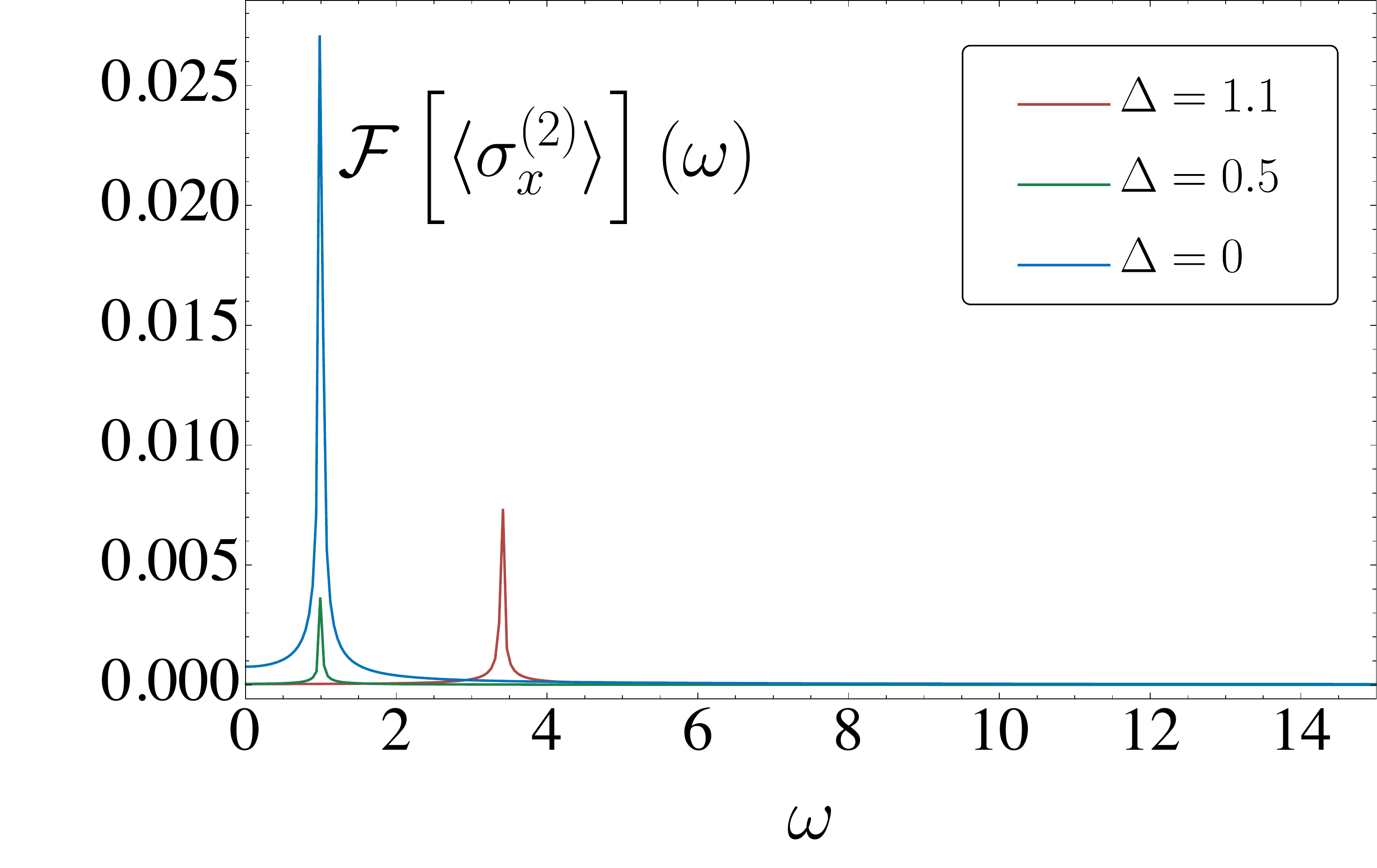}};
    \node[above=of img5, node distance=0cm, yshift=-4cm,xshift=2.1cm] {{\color{black}{\bf{(e)}}}};
\end{tikzpicture}
\begin{tikzpicture} 
  \node (img6)  {\includegraphics[width=0.33\linewidth]{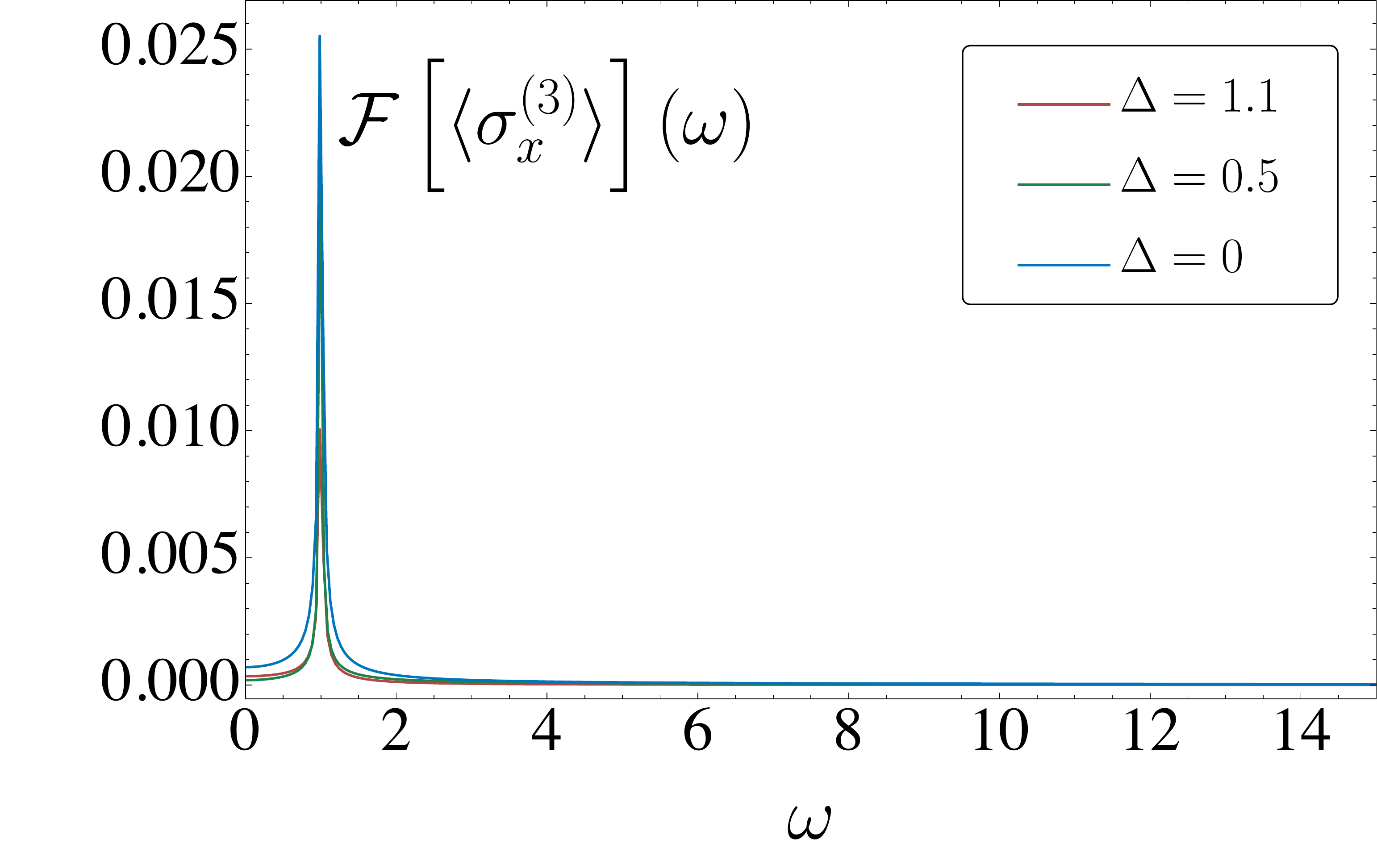}};
    \node[above=of img6, node distance=0cm, yshift=-4cm,xshift=2.1cm] {{\color{black}{\bf{(f)}}}};
\end{tikzpicture}
\caption{\textit{Upper panels:} Plots of $\langle\hat{\sigma}_x^{(k)}\rangle$, with (a) $k=1$ (i.e. site where the system-ancilla interaction takes place) (b) $k=2$ (i.e. adjacent site; due to symmetry reasons, the site $k=4$ shows the same behavior and thus is not plotted) and (c) $k=3$ (i.e. opposite site) for different values of the parameter $\Delta$, as reported in the legend, and for $\omega_0 = 1$, $\gamma = 0.5$ and $\tau = 1$. The respective insets show a zoom of these oscillations in the long-time limit, which are present for $k = 2,3$. 
\textit{Lower panels:} Plots of the absolute value of the corresponding long-time limit spectra $\mathcal{F}\left[\langle\hat{\sigma}_x^{(k)}\rangle\right] = \lvert \int_{\mathbb{R}} dt e^{-i\omega t} \langle\hat{\sigma}_x^{(k)}(t)\rangle \Theta(t-t^*)\rvert$, with $\Theta(t-t^*)$ being the Heaviside step function and $t^* = 250$ being an arbitrarily chosen large value of time after which the above oscillations, when present, have appeared. The peaks in panel (e) are at frequency $\omega=\omega_0-4\Delta$, while the peaks in panel (f) are at $\omega=\omega_0$. }
\label{figXXZchain}
\end{center}
\end{figure*}

An important consequence of Eq.~\eqref{theorem:Evol} is the onset of stable oscillations in generic (possibly local) system observables at long times. Here, the structure of repeated-interaction models, combined with Theorem~\ref{theorem1}, plays a key role. Using the fact that all quantum channels $\tilde{\Lambda}_{\tau,\theta}$ share the same peripheral spectrum~\footnote{The rest of the spectrum may be different due to the random nature of $\theta$.}, we demonstrate that taking the limit of many collisions $n \gg 1$ automatically singles out the oscillatory asymptotic states $\hat{\Xi}\hat{\omega}^D$ (see Supplementary Material, Section C, for a detailed proof~\cite{SM}). We emphasize that the above considerations hold true \textit{for generic choices of $\tau$ and $\theta$}.
As a consequence, the expectation value of any system observable $\hat{O}$ at time $t = n\tau + \sum_{j=1}^n \theta_j$, where $n\gg 1$, becomes 
\begin{align}\label{orderparameter}
    \mean{\hat{O}(t)} & \simeq r_0 \mathrm{Tr}\left[\hat{O} \,\hat{\omega}^D\right] \notag \\
    & \quad + \sum_\alpha e^{i\lambda_\alpha t} r_\alpha \mathrm{Tr}\left[ \hat{O} \, \hat{\Xi}_\alpha\,\hat{\omega}^D\right] + {\rm h.c.},
\end{align}
where $r_0 = \mathrm{Tr} \left[\hat{\omega}^D\rho_S^{(0)}\right]$, $r_\alpha = \mathrm{Tr} \left[\hat{\omega}^D_{\alpha}\rho_S^{(0)}\right]$ and $\hat{\omega}^D_\alpha \equiv \hat{\Xi}_\alpha\hat{\omega}^D$, and where we allow for multiple dynamical symmetries labelled by $\alpha$. Therefore, under the mild assumption that the initial state of the system and the chosen observable have non-zero overlap with at least one $\hat{\omega}^D_\alpha$, we find that oscillations spontaneously arise as a result of Eq.~\eqref{totalmap}, i.e. due to the randomized collisions with the ancillae. This periodicity manifests even at the level of a single realization of the stochastic process, provided that the many-body system under consideration features dynamical symmetries satisfying Theorem~\ref{theorem1}.

It is instructive to consider what happens when the waiting time $\theta$ takes a fixed, deterministic value. In the regime of ultra-fast collisions $\theta\ll\tau\ll 1$, it is known that the dynamics described by Eq.~\eqref{totalmap} becomes equivalent to a Markovian master equation in Lindblad form~\cite{GiovannettiPRL2012,VacchiniPRL2016,LorenzoPRA2017}. Furthermore, due to the constancy of $\theta$, the underlying dynamics Eq.~\eqref{totalmap} acquires a discrete time-translation symmetry, with period $\tau+\theta$.
In Section B of the Supplementary Material~\cite{SM} we explicitly show that conditions (i) and (ii) of Theorem~\ref{theorem1} reduce in this limit to the set of conditions first derived in~\cite{buvca2019non} that guarantee the spontaneous time-translation symmetry breaking and the manifestation of dissipative time-crystals for Lindblad dynamics. 

\textit{Example: the XXZ spin ring.---}We now demonstrate our results with an application to physically relevant many-body system. Consider a uniform Heisenberg XXZ spin chain with four sites arranged in a one-dimensional ring~\cite{takahashi2005thermodynamics, mendoza2013heat,BucaProsen2012}. The Hamiltonian is given by
\begin{align}
    \hat{\Ham}_S &= \sum_{j=1}^{M} \left[\hat{\sigma}_x^{(j)} \hat{\sigma}_x^{(j+1)} + \hat{\sigma}_y^{(j)} \hat{\sigma}_y^{(j+1)}\right] \notag\\
    &\quad+ \sum_{j=1}^{M} \left[\Delta \hat{\sigma}_z^{(j)} \hat{\sigma}_z^{(j+1)} + \frac{\omega_0}{2}\hat{\sigma}_z^{(j)}\right],
\end{align}
where $\hat{\sigma}_{x,y,z}^{(k)} \equiv \hat{\Id}^{k-1} \otimes \hat{\sigma}_{x,y,z}\otimes \hat{\Id}^{M-k} $ denotes a Pauli matrix acting on site $(k)$ of the chain, with periodic boundary conditions imposed, $\hat{\sigma}_{x,y,z}^{(M+1)} \equiv \hat{\sigma}_{x,y,z}^{(1)}$. The above system interacts with a stream of identical ancillae through a repeated interaction scheme as in Fig.~\ref{figschematic}. Specifically, the system-ancilla Hamiltonian is $\hat{\Ham}_{SA} = g \left(\hat{\sigma}_x^{(A)}\hat{\sigma}_x^{(1)} + \hat{\sigma}_y^{(A)}\hat{\sigma}_y^{(1)}\right)$,
where $g = \sqrt{\Gamma/4\tau}$. Finally, the ancillae are assumed to be initialized in the ground state $\hat{\rho}_A = \ket{0}^{(A)}\bra{0}^{(A)}$, where the standard computational basis states $\ket{0}^{(j)}$ and $\ket{1}^{(j)}$ denote the eigenvectors of $\hat{\sigma}_z^{(j)}$. For $M=4$, we find that there exists two dynamical symmetries satisfying conditions (i) and (ii) of Theorem~\ref{Theorem1} of the following form. The first, corresponding to the eigenvalue $\lambda_1 = -\omega_0$, is
\begin{align}
\label{Xi1}
&\hat{\Xi}_1 = \hat{\Id} \otimes \ket{\psi_1}\bra{\phi_1}, \nonumber \\
&\ket{\psi_1}=|0\rangle^{(2)}|0\rangle^{(3)}|1\rangle^{(4)}-|1\rangle^{(2)}|0\rangle^{(3)}|0\rangle^{(4)}, \nonumber \\
& \ket{\phi_1}=|0\rangle^{(2)}|1\rangle^{(3)}|1\rangle^{(4)}-|1\rangle^{(2)}|1\rangle^{(3)}|0\rangle^{(4)},
\end{align}
 with $\hat{\Id}$ being the identity at site 1, the site which couples to the ancilla. The second, corresponding to eigenvalue $\lambda_2 =\omega_0-4\Delta$,  is 
 \begin{align}
 \label{Xi2}
& \hat{\Xi}_2 = \ket{\psi_2}\bra{\phi_2},     \nonumber \\
& \ket{\psi_2}=\frac{1}{2}\left(|0\rangle^{(1)}|1\rangle^{(2)}|0\rangle^{(3)}|0\rangle^{(4)}-|0\rangle^{(1)}|0\rangle^{(2)}|0\rangle^{(3)}|1\rangle^{(4)}\right) \nonumber \\
& \ket{\phi_2}=|0\rangle^{(1)}|0\rangle^{(2)}|0\rangle^{(3)}|0\rangle^{(4)}.
 \end{align}

The existence of two distinct dynamical symmetries is reflected in the rich behavior of local observables on different sites of the ring. Figure~\ref{figXXZchain} shows the time evolution of $\mean{\hat{\sigma}_x^{(k)}(t)}$, for $k=1,2,3$, and their corresponding long-time Fourier spectra (in absolute value), with $\omega_0 = 1$, $\tau = 1$, and various values of $\Delta$. We take the time interval between collisions to be drawn from an exponential distribution $p(\theta) = \gamma e^{-\gamma\theta}$, with $\gamma=0.5$. In addition, we choose the initial density matrix $\hat{\rho}_S^{(0)}$ to be a random pure state. Remarkably, despite this random initial condition and random collisional dynamics, we observe stable oscillations emerging predictably at the level of single stochastic realizations and in all parameter regimes considered.

Spin observables on the site where the system-ancilla interaction takes place have zero overlap --- in the Hilbert-Schmidt sense of Eq.~\eqref{orderparameter} --- with both the dynamical symmetries $\hat{\Xi}_{1,2}$ and consequently $\mean{\hat{\sigma}_x^{(1)}(t)}$ shows no oscillation for any choice of parameters (see Fig.~\ref{figXXZchain}\textbf{(a)}). This behavior is also reflected in the absence of any systematic peak in the long-time spectrum, Fig.~\ref{figXXZchain}\textbf{(d)}). However, periodicity emerges for
observables defined locally on one of the adjacent sites (sites $2$ and $4$, for symmetry reasons, behave in the exact same way; see Fig.~\ref{figXXZchain}\textbf{(b)}) as well as the opposite site (i.e.~site $3$, see Fig.~\ref{figXXZchain}\textbf{(c)}).
Since the former have non-zero overlap with $\hat{\Xi}_2$, the resulting oscillations have frequency $\lambda_2 = \omega_0 - 4\Delta$, as evident from the Fourier analysis, see Fig.~\ref{figXXZchain}\textbf{(e)}. It is therefore possible to tune the oscillation period by changing the asymmetry parameter $\Delta$ of the many-body system. Conversely, $\hat{\sigma}_x^{(3)}$ only has overlap with the other dynamical symmetry $\hat{\Xi}_1$ and thus the resulting oscillations become insensitive to $\Delta$, being affected only by the on-site energy $\omega_0$ (see Fig.~\ref{figXXZchain}\textbf{(f)}). Moreover, since $\hat{\Xi}_1$ has identity at the site which couples to the ancilla (see Eq.(\ref{Xi1})), it is unaffected by any non-perturbative change of system-ancilla interaction. As a result, the oscillations in $\mean{\hat{\sigma}_x^{(3)}(t)}$ are \textit{non-perturbatively} robust against modifications of the system-ancilla interaction Hamiltonian~\cite{spinlace}.


To provide further evidence of the generality of this phenomenon, in the Supplemental Material we consider a long-range Ising chain with a similar eigenmode structure to the XXZ model, which has been realised in recent ion-trap experiments~\cite{Kim2009,Islam2013,Jurcevic2014,Jurcevic2015}. Despite its lack of translation invariance or total spin conservation, the long-range Ising model is also found to exhibit very stable oscillations under random repeated interactions. We note that the collisional dynamics could be implemented by a single qubit that is reset to the ground state after each interaction, e.g.~via optical pumping in an ion-trap setting.

\textit{Conclusions.---}In this Letter, we have demonstrated the onset of periodic time dynamics in a quantum many-body system even though the underlying evolution is random, and therefore has neither discrete nor continuous time translation symmetry. In a sense, this represents the opposite of a time crystal~\cite{Wilczek2012} because the long-time behavior of observables --- being invariant under discrete translations by the oscillation period --- has \textit{more} symmetry than the governing equations of motion, not less. The emergence of a stable time reference from timeless resources is a problem of perennial interest in quantum physics, with relevance for cosmology~\cite{Page1983}, thermodynamics~\cite{Erker2017,Schwarzhans2021}, and precision measurement~\cite{Ludlow2015}. Our work shows that dynamical symmetries provide a natural mechanism for this temporal regularity to arise from randomness in finite-dimensional quantum systems.


The generality of our formal result and the versatility of the repeated-interaction framework immediately suggest several directions for future research. These include optimization strategies to produce oscillations in certain target observables, and a careful analysis of the thermodynamic cost of maintaining a stable time-periodic state.

\textit{Acknowledgments.---}We acknowledge support from the European Research Council Starting Grant ODYSSEY (G. A. 758403).
GG kindly acknowledges funding from from European Unions Horizon 2020 research and innovation programme under the
Marie Sk\l{}odowska-Curie grant agreement No. 101026667, from FQXi and DFG FOR2724. MTM and JG acknowledge support from the EPSRC-SFI Joint Funding of Research grant QuamNESS. BB and DJ acknowledge funding from EPSRC programme grant EP/P009565/1, EPSRC National Quantum Technology Hub in Networked Quantum Information Technology (EP/M013243/1), and the European Research Council under the European Union's Seventh Framework Programme (FP7/2007-2013)/ERC Grant Agreement no. 319286, Q-MAC. JG is supported by a SFI-Royal Society University Research Fellowship. AP acknowledges funding from European Union’s Horizon 2020 research and innovation programme under the Marie Sk\l{}odowska-Curie grant agreement No. 890884. BB is grateful to M. Medenjak for numerous useful discussions on the lossy XXZ spin ring and collaboration on related subjects.
GG is grateful to J. Eisert and A. Riera-Campeny for insightful discussions on time-crystals in many-body open quantum systems. We thank F. Schmidt-Kaler for suggesting the long-range Ising model.

\bibliographystyle{apsrev4-1}
\bibliography{bibliographyCTC}

\widetext
\clearpage
\begin{center}
\textbf{\large Supplementary Material}\\
\end{center}
\setcounter{equation}{0}
\setcounter{figure}{0}
\setcounter{table}{0}
\setcounter{page}{1}
\makeatletter
\renewcommand{\theequation}{S\arabic{equation}}
\renewcommand{\thefigure}{S\arabic{figure}}

In this Supplementary Material we provide all the detailed calculations of the results presented in the main text, as well as some additional considerations concerning the models discussed.
Note that throughout the paper we set $\hbar = 1$ and $k_B=1$.
The organization is as follows: in Section~\ref{Sec:Notation} we provide some basic notions of open quantum systems theory which help clarifying the notation used throughout this paper. Section~\ref{Sec:Proofs} contains the rigorous proofs of Theorem 1 and a following Corollary. Section~\ref{App:CPTPmapsproperties} includes important additional considerations regarding spectral properties of generic completely-positive and trace-preserving (CPTP) maps; it furthermore provides the rigorous proof that the sequential application of CPTP maps sharing the same peripheral spectrum allows to single out the eigenvectors generating the latter. This consideration, in combination with Theorem 1, lead to the main result of this paper.
Finally, this Supplementary Material concludes with Section~\ref{Sec:IsingModel} which is dedicated to the analysis of the long-range Ising model, an additional example recently implemented in ion-trap architectures.

\section{A --- Mathematical preliminaries and useful notation.}
\label{Sec:Notation}

Let us consider an open quantum system consisting of a system $S$ interacting with an environment $E$. The total Hamiltonian is given by $\hat{\Ham} = \hat{\Ham}_S + \hat{\Ham}_E + \hat{\Ham}_{SE}$. Assume moreover that the system and the bath are initially prepared in a factorized state, i.e. $\hat{\rho}_{SE}(0) = \hat{\rho}_S(0)\otimes\hat{\rho}_E$. Then, if we denote with $\tau(\mathscr{H})$ the set of trace 1 positive linear operators, i.e. density matrices, on the Hilbert space $\mathscr{H}$, the total unitary evolution dictating the evolution of the full composite system induces a \textit{completely positive and trace-preserving} (CPTP) map $\Lambda_t \, : \tau(\mathscr{H}_S) \longrightarrow \tau(\mathscr{H}_S)$ which can be expressed in Kraus form as
\begin{equation}
\Lambda_t \left(\hat{\rho}_S\right) = \sum_k \hat{\Omega}_k(t) \rho_S \hat{\Omega}_k^{\dagger}(t),
\end{equation}
where $\hat{\Omega}_k(t) \equiv \hat{\Omega}_{\alpha\beta}(t) = \sqrt{p_\alpha}\bra{\beta} \hat{U}(t) \ket{\alpha}$, $\hat{U}(t) = e^{-it\hat{\Ham}}$ and $\lbrace\ket{\alpha}\rbrace$ denotes an orthonormal basis of $\mathscr{H}_E$ such that $\hat{\rho}_E = \sum_{\alpha} p_{\alpha}\ket{\alpha}\bra{\alpha}$.
Local observables of the system are represented by self-adjoint bounded linear operators on $\mathscr{H}_S$ and form a set denoted by  $ \mathcal{B}(\mathscr{H}_S) $. 
The trace induces a duality form between $ \tau (\mathscr{H}_S ) $ and $ \mathcal{B}(\mathscr{H}_S ) $
\begin{align}\label{eq:HSscalarproduct}
\mathrm{Tr}: &\mathcal{B}(\mathscr{H}_S ) \times \tau (\mathscr{H}_S ) \rightarrow \mathbb{R} \notag\\
& (\hat{O}, \hat{\rho}) \mapsto\mathrm{Tr}\left[\hat{O}^{\dagger} \hat{\rho}\right]
\end{align}
which can be used to define a scalar product on $ \tau (\mathscr{H}_S) $, known as \emph{Hilbert-Schmidt} scalar product. The latter is widely employed in the so-called \textit{vectorization} procedure, which maps operators to vectors and maps (superoperators) to matrices.
To make things explicit, consider an orthonormal basis $ \lbrace \hat{\sigma}_{\alpha} \rbrace $ with respect to the Hilbert-Schmidt scalar product \eqref{eq:HSscalarproduct}, i.e.$\mathrm{Tr}\left[\hat{\sigma}_\alpha^\dagger\hat{\sigma}_\beta\right] = \delta_{\alpha,\beta}.$
Each operator $O$ can thus be expanded as
\begin{equation}\label{eq:OpDecomposition}
    \hat{O} = \sum_{\alpha} O_{\alpha} \hat{\sigma}_{\alpha}, \qquad O_{\alpha} =  \mathrm{Tr}\left[\hat{\sigma}_{\alpha}^{\dagger} \hat{O}\right] \in\mathbb{R}.
\end{equation}
The entries $O_\alpha$ can then be collected in a vector $\mathbf{O}$.
Similarly, each linear map $ \Lambda_t $ acting on a generic $\hat{\rho}\in\tau (\mathscr{H}_S)$ can be decomposed on this basis
\begin{equation}\label{eq:MapDecomposition}
\Lambda_t[\hat{\rho}] = \sum_{\alpha ,\beta} \Lambda_{\alpha \beta}(t) \mathrm{Tr}\left[\hat{\sigma}^{\dagger}_{\beta}\, \hat{\rho}\right] \hat{\sigma}_{\alpha},\quad  \Lambda_{\alpha \beta}(t) = \mathrm{Tr}\left[\hat{\sigma}^{\dagger}_{\alpha}\,\Lambda_t[\hat{\sigma}_{\beta}]\right]\in\mathbb{R}.
\end{equation}
The numbers $\Lambda_{\alpha \beta}(t)$ can be rearranged into a matrix $\bgreek{\Lambda}(t)$.
It is worth stressing that the particular Hilbert-Schmidt representation of an operator or superoperator depends on the choice of the basis while the expansions Eqs.~\eqref{eq:OpDecomposition}~\eqref{eq:MapDecomposition}. The vectorization procedure is widely employed, together with the fundamental relation that expresses linear transformations in terms of matrix multiplications: given three generic operators $\lbrace \hat{O}_i\rbrace_{i=1,2,3} \in\mathcal{B}(\mathscr{H}_S)$
\begin{equation}
    \hat{O}_1 \hat{O}_2 \hat{O}_3 \mapsto \left(\mathbf{O}_3^T\otimes\mathbf{O}_1\right)\mathbf{O}_2,
\end{equation}
with the superscript $`T`$ denoting the transpose.
A straightforward application is readily provided: given a Kraus decomposition of dynamical map, one can find its matrix representation as
\begin{equation}\label{eq:Kraustomap}
    \bgreek{\Lambda}_t = \sum_k \bgreek{\overline{\Omega}}_k(t) \otimes \bgreek{\Omega}_k(t),
\end{equation}
with $\bgreek{\overline{\Omega}}(t)$ denoting the complex conjugate of $\bgreek{\Omega}(t)$.

Finally, it is useful to introduce the concept of \textit{dual map} $\Lambda^{\ddagger}_t$ (sometimes also known as adjoint map). The latter is defined by the following relation
\begin{equation}\label{adjoint}
\mathrm{Tr}\left[\Lambda^{\ddagger}_t[\hat{O}] \, \hat{\rho} \right] = \mathrm{Tr}\left[\hat{O} \,\Lambda_t[\hat{\rho}] \right) , \quad \forall \hat{O} \in \mathcal{B}(\mathscr{H}_S), \hat{\rho} \in \tau(\mathscr{H}_S).
\end{equation}
As it is easy to check, the Hilbert-Schmidt matrix $\bgreek{\Lambda^{\ddagger}}(t)$ associated to the dual map coincides with the hermitian conjugate (i.e. conjugate transpose) of the matrix $\bgreek{\Lambda}(t)$, i.e.
\begin{equation}\label{adjointmatrix}
    \bgreek{\Lambda^{\ddagger}}(t) = \bgreek{\Lambda}^{\dagger}(t).
\end{equation}.

\section{B -- Proofs of the Theorem and Corollaries.}\label{Sec:Proofs}

In what follows, the symbols and quantities refer to the ones introduced in Section~\ref{Sec:Notation}.
It will be useful to introduce the interaction picture with respect to the system's free Hamiltonian, which allows to `separate` the purely coherent evolution generated by $\hat{\Ham}_S$ and the non-unitary part induced by the interaction with the environment. For this purpose, we will adopt the following notation
\begin{equation}
\Lambda_\tau [\bullet] = \mathcal{U}_{S,\tau} \circ \Lambda_{I,\tau}[\bullet] , 
\end{equation}
where $\mathcal{U}_{S,\tau} \left[\bullet \right] \equiv \hat{U}_S(\tau) \bullet \hat{U}_S^{\dagger}(\tau)$ (with $\hat{U}_S(\tau) = e^{-i \tau \hat{\Ham}_S}$) denotes the free system evolution superoperator and where $\Lambda_{I,\tau} [\bullet] = \sum_k \hat{\Omega}_{k,I}(\tau) \bullet \hat{\Omega}_{k,I}^{\dagger}(\tau) $ is the quantum channel in interaction picture with respect to the system's free evolution, with $\Omega_{k,I}(\tau) \equiv \hat{\Omega}_{\alpha\beta}(\tau) = \sqrt{p_\alpha} \bra{\beta} \hat{U}_I(\tau) \ket{\alpha}$, with $\hat{U}_I(\tau) = \hat{U}^{\dagger}_S(\tau)\hat{U}(\tau)$ and  $\hat{\rho}_A = \sum_\alpha p_\alpha\ket{\alpha}\bra{\alpha}$.
Finally, in light of the main focus of the present work, we will be interested in characterising the CPTP map
$\Lambda_{t}[\bullet] = \mathcal{U}_{S,\theta}\circ\Lambda_\tau[\bullet]$, where $t = \tau + \theta$.

\begin{theorem*}\label{Theorem1}
Let us consider a CPTP map $\tilde{\Lambda}_{\tau,\theta} [\bullet] = \mathcal{U}_{S,\theta}\circ\Lambda_\tau[\bullet]$and let $\hat{\omega}^D$ be its invariant state such that $\tilde{\Lambda}_{\tau,\theta}[\hat{\omega}^D] = \hat{\omega}^D$.
If there exists a system operator $\hat{\Xi}$ such that the following two conditions are satisfied
\begin{equation}
\mathrm{(i)}\, \left[\hat{\Ham}_S, \hat{\Xi}\right] = \lambda \,\hat{\Xi}, \qquad \mathrm{(ii)}\, \left[\hat{\Omega}_{k,I}(\tau), \hat{\Xi}\right] \hat{\omega}^D = 0,\,\forall k,\tau
\end{equation}
then the operator will evolve according to
\begin{equation}\label{Theorem:Evol}
\tilde{\Lambda}_{\tau,\theta}[\hat{\Xi} \hat{\omega}^D] = e^{-i\lambda (\tau+\theta)} \,\hat{\Xi}\,\hat{\omega}^D,\,\, \forall t,
\end{equation}
with $\lambda\in\mathbb{R}$.
\end{theorem*}

\begin{proof}
First of all, let us consider the adjoint of condition $(i)$, which reads $ \left[\hat{\Ham}_S, \hat{\Xi}^{\dagger}\right]  = -\lambda^* \hat{\Xi}^{\dagger}$. Then the superoperator $\Psi\left(\bullet\right) \equiv \left[\hat{\Ham}_S, \hat{\Xi}^{\dagger}\hat{\Xi}\right]\bullet $ is clearly skew-hermitian, being a commutator of two hermitian operators and consequently its spectrum is purely imaginary. This in particular implies that the Hilbert-Schmidt inner product of any self-adjoint operator $\hat{A}=\hat{A}^{\dagger}\in\mathcal{B}(\mathscr{H}_S)$, one has that $\langle\langle \hat{A}\Psi (\hat{A})\rangle\rangle \equiv \mathrm{Tr}_{S} \, \hat{A} \Psi (\hat{A}) = i\theta $, $\theta\in\mathbb{R}$. On the other hand, by exploiting the above condition $(i)$, one finds that
\begin{align}
\langle\langle \hat{\Id} \Psi (\hat{\Id})\rangle & =\mathrm{Tr}_{S} \, \hat{\Id} \left[\hat{\Ham}_S, \hat{\Xi}^{\dagger}\hat{\Xi}\right] \hat{\Id} \notag\\
&= \mathrm{Tr}_{S}\, \hat{\Id}\left(\hat{\Xi}^{\dagger}\left[\hat{\Ham}_S, \hat{\Xi}\right] +  \left[\hat{\Ham}_S, \hat{\Xi}^{\dagger}\right]\hat{\Xi} \right)\hat{\Id} \notag\\
&\overset{(i)}{=}\mathrm{Tr}_{S}\, \hat{\Id} \hat{\Xi}^{\dagger} (\lambda \hat{\Xi} \hat{\Id}) - \lambda^*\hat{\Id} \hat{\Xi}^{\dagger} \hat{\Xi} \hat{\Id}\notag\\
&= (\lambda - \lambda^*) \mathrm{Tr}_{S}\, \hat{\Xi}^{\dagger} \hat{\Xi} \notag\\
&= (\lambda - \lambda^*)  ||\hat{\Xi}||_1.
\end{align}
Since the trace-norm of any product of bounded operator and density matrix is always $\geq 0$, this implies that $\lambda = \lambda^*$, i.e. $\lambda \in \mathbb{R}$.

Let us now calculate $\tilde{\Lambda}_{\tau,\theta}[\hat{\Xi}\hat{\omega}^D]$. Straightforward algebra then allows to show that
\begin{align}\label{firststepproof}
    \tilde{\Lambda}_{\tau,\theta}[\hat{\Xi}\hat{\omega}^D]  &\equiv \mathcal{U}_{S,\theta}\circ\Lambda_{\tau}[\hat{\Xi}\hat{\omega_D}] = \mathcal{U}_{S,\theta}\circ\mathcal{U}_{S,\tau}\circ\Lambda_{I,\tau}[\hat{\Xi}\hat{\omega_D}] = \mathcal{U}_{S,\theta+\tau}\circ\Lambda_{I,\tau}[\hat{\Xi}\hat{\omega_D}]\notag\\
    &=\hat{U}_S(\tau+\theta) \Lambda_{I,\tau}\left[\hat{\Xi}\,\hat{\omega}^D\right] \hat{U}_S^{\dagger}(\tau+\theta) = \hat{U}_S(\tau+\theta) \left( \sum_k \hat{\Omega}_{k,I}(\tau) \hat{\Xi}\,\hat{\omega}^D \hat{\Omega}_{k,I}^{\dagger}(\tau) \right) \hat{U}_S^{\dagger}(\tau+\theta) \notag\\
    &\overset{(ii)}{=} \hat{U}_S(\tau+\theta) \left( \hat{\Xi} \sum_k \hat{\Omega}_{k,I}(\tau) \hat{\omega}^D \hat{\Omega}_{k,I}^{\dagger}(\tau) \right) \hat{U}_S^{\dagger}(\tau+\theta)\notag\\
    &= \hat{U}_S(\tau+\theta) \hat{\Xi} \hat{U}_S^{\dagger}(\tau+\theta) \hat{U}_S(\tau+\theta)\left(\Lambda_{I,\tau}[\hat{\omega}^D]\right)\hat{U}_S^{\dagger}(\tau+\theta)\notag\\
    &= \mathcal{U}_{S,\tau+\theta}\left[\hat{\Xi}\right] \mathcal{U}_{S,\tau+\theta}\circ\Lambda_{I,\tau}[\hat{\omega}^D] = \mathcal{U}_{S,\tau+\theta}\left[\hat{\Xi}\right] \tilde{\Lambda}_{\tau,\theta}[\hat{\omega}^D]\notag\\
    &= \mathcal{U}_{S,\tau+\theta}\left[\hat{\Xi}\right] \hat{\omega}^D,
\end{align}
where in the first line the group composition law of unitary superoperators $\mathcal{U}_{S,\tau}\circ\mathcal{U}_{S,\theta} = \mathcal{U}_{S,\tau+\theta}$ has been used, the third line was obtained from the second by making use of (ii), an identity $\hat{\Id} = \hat{U}_S^{\dagger} \hat{U}_S $ was inserted in the fourth line and the property of $\hat{\omega}^D$ to be the invariant state of the composite map $\tilde{\Lambda}_{\tau,\theta}$ was used to get the last line.

To complete the proof, one needs to compute $\mathcal{U}_{S,\tau+\theta}\left[\hat{\Xi}\right]$. In order to do that, let us for brevity denote $t \equiv \tau+\theta$ and let us consider the time derivative of the above expression with respect to $t$. By making use of the explicit expression for $\hat{U}_S(t)$, i.e. $\hat{U}_S(t) = e^{-it\hat{\Ham}_S}$, as well as the trivial relation $\left[\hat{U}_S(t), \hat{\Ham}_S\right] = 0$, one obtains
\begin{align}\label{eq:almost}
    &\frac{d \mathcal{U}_{S,t}\left[\hat{\Xi}\right]}{dt} \notag\\
    & \; = \hat{U}_S(t) \left(-i\hat{\Ham}_S\right)\hat{\Xi} \hat{U}_S^{\dagger}(t) + \hat{U}_S(t) \hat{\Xi}\left(i\hat{\Ham}_S\right) \hat{U}_S^{\dagger}(t)\notag\\
    & \;=  -i \hat{U}_S(t) \left[\hat{\Ham}_S, \hat{\Xi}\right] \hat{U}_S^{\dagger}(t)\notag\\
    & \;\overset{(i)}{=} -i \hat{U}_S(t) \lambda \hat{\Xi}\hat{U}_S^{\dagger}(t)\notag\\
     & \;= -i\lambda\, \mathcal{U}_{S,t}\left[\hat{\Xi}\right],
\end{align}
where condition (i) was used to obtain the second-to-last line.
Let us now introduce the operator $\hat{\sigma}(t) \equiv \mathcal{U}_{S,t}\left[\hat{\Xi}\right]$, which obviously satisfies the boundary condition $\hat{\sigma}(0) = \hat{\Xi}$. Then, Eq.~\eqref{eq:almost} can be compactly expressed as 
\begin{align}
 &\frac{d}{dt}\hat{\sigma}(t) = -i\lambda \hat{\sigma}(t) \notag\\
 &\hat{\sigma}(0) = \hat{\Xi},
\end{align}
which leads to the solution
\begin{equation}\label{eq:result}
    \hat{\sigma}(t) \equiv \mathcal{U}_{S,t}\left[\hat{\Xi}\right] = e^{-i\lambda t}\hat{\Xi}.
\end{equation}
Inserting Eq.~\eqref{eq:result} into Eq.~\eqref{firststepproof} and re-expressing $t = \tau+\theta$ finally allows to demonstrate that
\begin{equation}
\tilde{\Lambda}_{\tau,\theta}\left[\hat{\Xi}\hat{\omega}^D\right] = e^{-i\lambda (\tau+\theta)}\hat{\Xi}\,\hat{\omega}^D.
\end{equation}
\end{proof}

Here below we furthermore study which form conditions (i) and (ii) of the previous Theorem take in the limiting regime where the randomized collision model dynamics described by Eq.~\eqref{totalmap} reduces to a Lindblad master equation.
For the latter to be valid, one needs to consider the case where the collisions stop happening at random times but occur periodically after a time $\theta$, which then becomes a fixed quantity. This implies that the underlying dynamics described by this collision model acquires a time-periodic symmetry, i.e. the dynamics is given in terms of a repeated interaction scheme with period $\tau+\theta$. For simplicity and without loss of generality, let us consider the case where $\theta = 0$; the case where $\theta$ is a finite is straightforwardly obtained by a simple change of variable $\tau \to \tau+\theta$.

\begin{corollary}
Let us consider a state $\hat{\omega}^D$ and an operator $\hat{\Xi}$ such that they satisfy the conditions in Theorem 1. Let us moreover assume that the interaction time $\tau$ is vanishing small while the system-environment interaction Hamiltonian scales as $\hat{\Ham}_{SE}(\tau) =\tau^{-1/2} \hat{V}_{SE}$, and that the initial state of the environment $\hat{\rho}_E = \sum_{\alpha} p_{\alpha}\ket{\alpha}\bra{\alpha}$ satisfies the KMS condition ~\cite{BreuerPetruccione}. Then, the conditions (i) and (ii) of Theorem 1 become
\begin{itemize}
\item[(i')] $  \left[\hat{\Ham}_S,\hat{\Xi}\right] = \lambda \hat{\Xi}, $
\item[(ii')] $\left[\hat{L}_k(\tau),\hat{\Xi}\right] \hat{\omega}^D = \left[\hat{L}^{\dagger}_k(\tau),\hat{\Xi}\right] \hat{L}_k(\tau) \hat{\omega}^D  0, \qquad \forall k$,
\end{itemize}
where $\hat{L}_k(\tau) \equiv \hat{L}_{\alpha\beta}(\tau) = \sqrt{p_\alpha/t} \bra{\beta} \hat{V}_{SE} \ket{\alpha}$ is a Lindblad operator. 
\end{corollary}

\begin{proof}
Being a stationary state of the dynamics, $\hat{\omega}^D$ satisfies the following equation 
\begin{equation}\label{eq:map}
\hat{\omega}^D = \Lambda_\tau(\hat{\omega}^D) = \mathrm{Tr}_E\, \hat{U}(\tau) \hat{\omega}^D\otimes\hat{\rho}_E \hat{U}^{\dagger}(\tau),
\end{equation}
where $\hat{U}(\tau)$ denotes the unitary evolution operator acting on the composite system.
Performing the above mentioned limit $\lim_{\tau\to 0^+}$ means Taylor-expanding the unitary evolution operator in powers of $\tau$; exploiting the particular scaling of the interaction Hamiltonian  and retaining only the terms up to first order in $t$ leads to
\begin{equation}\label{expansion}
\hat{U}(\tau)  = \hat{\Id} + i \sqrt{\tau} \hat{V}_{SE} - \tau \left( i\hat{\Ham}_S+i\hat{\Ham}_E+\hat{V}_{SE}\right).
\end{equation}
The choice of initial state for the environment, i.e. one that satisfies the KMS condition (e.g. a thermal state) guarantees that, without loss of generality, we can set $\mathrm{Tr}_E \hat{V}_{SE}\hat{\rho}_E = 0$. By exploiting this fact,  inserting Eq.~\eqref{expansion} into Eq.~\eqref{eq:map} then leads after a little algebra to
\begin{equation}\label{eq:Lindblad}
\mathcal{D}_\tau(\hat{\omega}^D) = -i\left[\hat{\Ham}_S,\hat{\omega}^D\right],
\end{equation}
where we introduced the superoperator $\mathcal{D}_\tau(\hat{\rho}) \equiv \sum_k \left(\hat{L}_k(\tau) \hat{\rho} \hat{L}^{\dagger}_k(\tau) - \frac{1}{2}\lbrace \hat{L}^{\dagger}_k(\tau) \hat{L}_k(\tau),\hat{\rho}\rbrace\right)$ and where $\hat{L}_k(\tau) \equiv \hat{L}_{\alpha\beta}(\tau) = \sqrt{p_\alpha/\tau} \bra{\beta} \hat{V}_{SE} \ket{\alpha}$. It is now immediate to recognize that Eq.~\eqref{eq:Lindblad} corresponds to 
\begin{equation}\label{eq:peripheral}
\mathcal{L}_\tau(\hat{\omega}^D) = 0,
\end{equation}
with $\mathcal{L}_\tau$ being a superoperator in Gorini-Kossakowski-Sudarshan-Lindblad form. Eq.~\eqref{eq:peripheral} thus implies that $\hat{\omega}^D$ corresponds to the fixed point of the Lindblad superoperator and thus belongs to its so-called \textit{peripheral spectrum}.
By thus retracing the same steps performed in the Proof of Supplementary Theorem 1 of~\cite{buvca2019non}, it then follows that conditions (i) and (ii) translate in this limiting case to (i') and (ii'). 
\end{proof}

\section{C -- Additional considerations on the peripheral spectrum of CPTP maps and their composition.}
\label{App:CPTPmapsproperties}

The CPTP nature of a dynamical map $\Lambda$ allows to draw general insightful considerations on its spectrum, of relevance for the present work. First of all, its complete positivity guarantees that it is always at least ``Jordan-diagonalizable''. The set of all proper or generalized left and right eigenoperators, denoted respectively with $\lbrace\hat{\omega}_\mu\rbrace$ and $\lbrace\hat{\zeta}_\mu\rbrace$, forms a basis in the system's operators space $\mathcal{B}(\mathscr{H}_S)$ according to the Hilbert-Schmidt scalar product, i.e.
$\mathrm{Tr}\left[\hat{\zeta}_{\mu}\,\hat{\omega}_\nu \right] = \delta_{\mu\nu}$.
Due to the fact that every CPTP map is a \textit{contraction}, its complex eigenvalues satisfy the constraint $|\lambda_\mu| \leq 1$, i.e. they all lie either on the border or within the unit circle of the complex plane ~\cite{Wolf_2012,bruzda2009random,riera2020time}. 
Finally, being trace-preserving implies that the subset of invariant eigenoperators will always consist of at least one element.
Furthermore, if there exists at least one $\mu$ such that $(\lambda_\mu)^l=1$, with $l \in [2,+\infty)$, then the CPTP map has multiple invariant states. 
By having observed that these eigenvalues lie on the complex unit circle it is easy to finally realize that they always come in pair conjugates.

Assuming that the only persistently oscillating eigenmodes in the full map are those satisfying the conditions of Th.~\ref{Theorem1} for each map in the composition \eqref{totalmap}. In that case all the CPTP maps in the composition have the same persistently oscillating eigenmodes. Denote with $P$ a superprojector to this subspace and $Q=\one-P$.Then for any state $P \hat{\rho}=0$ we have with $\Lambda_j=\Lambda_\tau \mathcal{U}_{S,\theta}$,
\begin{equation}
\prod_j {\Lambda_j}\rho=\prod_j Q{\Lambda_j }Q \hat{\rho}.
\end{equation}
Moreover, $R_j=Q{\Lambda_j}Q$ all have eigenvalues strictly inside the unit circle because the spectral gap for each $\Lambda_j$ is non-zero. A finite spectral gap follows from the assumptions that 1. the Hilbert space is finite dimensional and, 2. as the gap should be a continuous function of the norm of collision CPTP part of the map $\Lambda_\tau$ and the collision times because roots of the characteristic polynomial are continuous functions of the matrix elements. In other words, if we keep $\Lambda_\tau$ and collision time finite the spectral gap is finite for each $ \Lambda_j$.  Therefore, the maximal singular value of each satisfies $Re[\sigma(Q{\Lambda_j}Q)]<1$, which is the 2-norm $\sigma$. Using sub-multiplicativity of the norm $||R_1 R_2...|| \le||R_1||||R_2...||=\le \sigma_1 ||R_2...||$ repeatedly we get that the norm of the product goes to 0. This implies that in the long-time limit the peripheral spectrum of the full map is the one given by the conditions of Th. 1 on each individual map. We note that this is the case in all the examples studied. 

\section{D -- The long-range Ising model example}\label{Sec:IsingModel}

As an example related to the one in the Main text we consider a one-dimensional Ising model with long-range interactions described by the Hamiltonian
\begin{equation}
    \label{Ising_Hamiltonian}
    \hat{\Ham}_S = \sum_{i=1}^M B \hat{\sigma}_z^{(i)} + \sum_{i<j} J_{ij} \hat{\sigma}^{(i)}_x \hat{\sigma}^{(j)}_x,
\end{equation}
where $B$ is an on-site magnetic field and exchange interactions are proportional to $J_{ij} = J/|i-j|^\alpha$ with the exponent $\alpha$ describing the distance dependence. Notably, this model is not translation-invariant and lacks the U(1) symmetry of the XXZ ring considered in the main text.  Nevertheless, this system shows similar physics, as we will show shortly. Importantly, it has been implemented in an ion-trap architecture and studied in a series of recent experiments~\cite{Kim2009,Islam2013,Jurcevic2014,Jurcevic2015}. In particular we will present the results of the simulations obtained by choosing the parameters from the experiments reported in Ref.~\cite{Jurcevic2015}, i.e. $M=7$, $\alpha = 1.1$, and $B=5J$. This system interacts through a randomized collision model scheme with a series of qubit ancillae coupled to the central site of the chain, via the interaction Hamiltonian 
\begin{equation}
    \label{Ising_ancilla}
    \hat{\Ham}_{SA} = B\hat{\sigma}_z^{(A)} + g\left( \hat{\sigma}_x^{(A)}\hat{\sigma}_x^{(4)} + \hat{\sigma}_y^{(A)}\hat{\sigma}_y^{(4)} \right),
\end{equation}
where $g=\sqrt{\Gamma/4\tau}$. The ancillae are once again initialised in the ground state. 

In Fig.~\ref{fig:Ising} we plot results for the odd/even population imbalance, $\mathcal{I} = \sum_{j \,\rm odd} \langle \hat{\sigma}_z^{(j)}\rangle - \sum_{j \, \rm even} \langle \hat{\sigma}_z^{(j)}\rangle$, as a function of time starting from a random, pure initial state. We find that the imbalance generically exhibits oscillations in the long-time limit. Interestingly, we find that the frequency of the oscillations is the same for each random preparation, but the phase and amplitude of the resulting oscillations differ. The random phase of the oscillations in the long-time limit may be understood to be analogous to the random orientation of the order parameter in the ordered phase of a symmetry-breaking phase transition. 

\begin{figure}
    \centering
    \begin{minipage}{0.48\linewidth}
        \includegraphics[width=\linewidth]{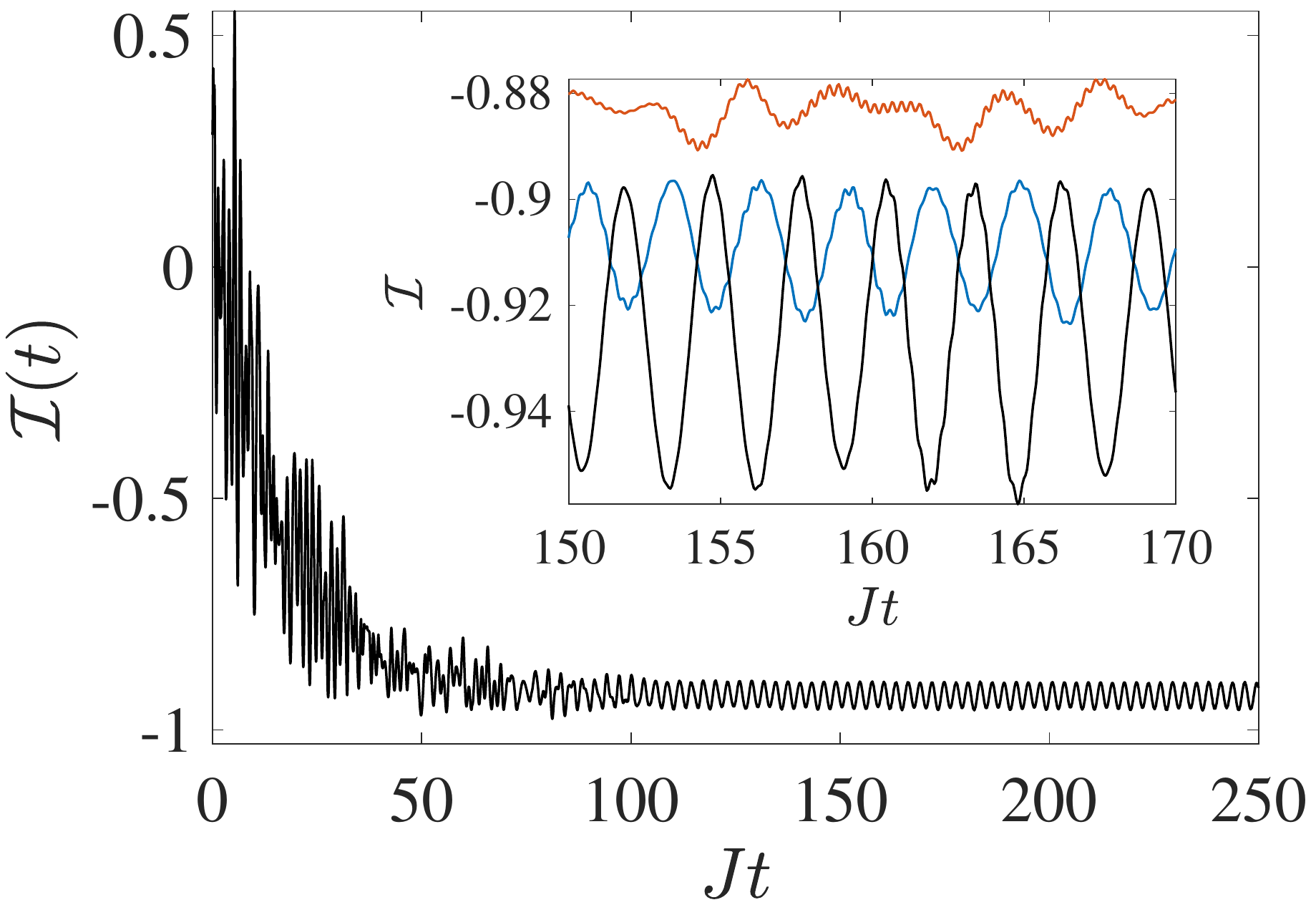}
    \end{minipage}
        \begin{minipage}{0.48\linewidth}
        \includegraphics[width=\linewidth]{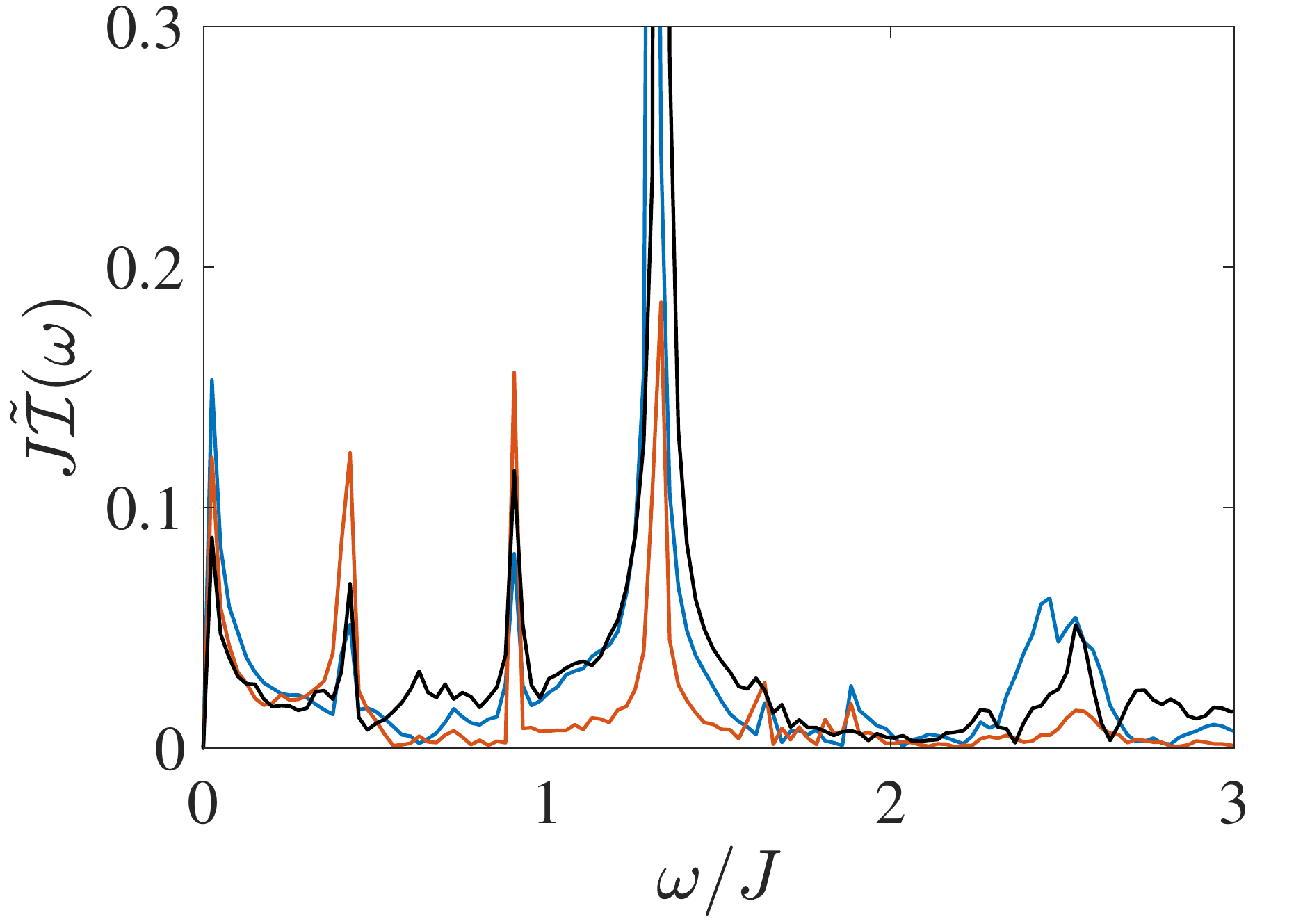}
    \end{minipage}
    \caption{(Left panel) Time evolution of the odd/even population imbalance for the long-range Ising model. The inset shows a close-up of the oscillations at long times, for three different random initial conditions. (Right panel) Frequency spectrum of the long-time imbalance, for the same initial conditions. A dominant Fourier peak is clearly visible. To evaluate the Fourier transform we take only the long-time data for $t>t^*$, with $Jt^*=100$. Parameters: $M=7$,  $B=5J$, $\Gamma = 4J$, $\tau J = 0.5$ and $\gamma \tau = 0.5$.}
    \label{fig:Ising}
\end{figure}

\end{document}